\documentclass[11pt]{article}
\usepackage[T1]{fontenc}
\usepackage{amsfonts}
\usepackage{amsmath}
\usepackage{amssymb}
\usepackage{amsthm}
\usepackage{bbm}
\usepackage{bm}
\usepackage{mathrsfs}
\usepackage{verbatim}
\usepackage{setspace}
\usepackage{color}
\usepackage{pdfsync}
\usepackage[round,authoryear]{natbib}
\usepackage{enumitem}
\theoremstyle{plain}
\newtheorem{theorem}{Theorem}[section]
\newtheorem{proposition}[theorem]{Proposition}
\newtheorem{lemma}[theorem]{Lemma}

\theoremstyle{definition}
\newtheorem{definition}[theorem]{Definition}
\newtheorem{remark}[theorem]{Remark}

\newtheorem{example}[theorem]{Example}

\theoremstyle{remark}

\renewenvironment{thebibliography}[1]{%
\begin{oldthebibliography}{#1}%
\setlength{\baselineskip}{.9em}
\linespread{1}
\small
\setlength{\parskip}{0.3ex}%
\setlength{\itemsep}{.5em}%
}%
{%
\end{oldthebibliography}%
}

\newcommand{\E}{\mathbb{E}}
\newcommand{\F}{\mathbb{F}}

\newcommand{\N}{\mathbb{N}}
\renewcommand{\P}{\mathbb{P}}
\newcommand{\Q}{\mathbb{Q}}
\newcommand{\R}{\mathbb{R}}

\newcommand{\cP}{\mathcal{P}}

\numberwithin{equation}{section}

\usepackage[pdfborder={0 0 0}]{hyperref}
\hypersetup{
  urlcolor = black,
  pdfauthor = {Ariel Neufeld},
  pdfkeywords = {T.B.D.;
},
  pdftitle = {Buy-and-Hold Property for Fully Incomplete Markets},
  pdfsubject = {Buy-and-Hold Property for Fully Incomplete Markets},
  pdfpagemode = UseNone
}

\begin{document}

\title{\vspace{-0em}
Buy-and-Hold Property for Fully Incomplete Markets when Super-replicating Markovian Claims
\author{
  Ariel Neufeld%
   \thanks{
   RiskLab, Department of Mathematics, ETH Zurich, \texttt{ariel.neufeld@math.ethz.ch}.
   Financial support by the Swiss National Foundation grant SNF 200020$\_$172815 is gratefully acknowledged.
   }
 }
}
\maketitle \vspace{-1.2em}

\begin{abstract}
We show that when the price process $S$ represents a fully incomplete market,  the optimal super-replication of any Markovian claim $g(S_T)$ with $g(\cdot)$ being nonnegative and lower semicontinuous is of buy-and-hold type. Since both (unbounded) stochastic volatility models and rough volatility models are examples of fully incomplete markets, one can interpret the buy-and-hold property when super-replicating  Markovian claims as a natural phenomenon in incomplete markets.
\end{abstract}

\vspace{.9em}

{\small
\noindent \emph{Keywords} Super-replication; Fully incomplete markets; Robust pricing;\\ $\phantom{\emph{Keywords}}$ Stochastic volatility; Rough volatility

\noindent \emph{AMS 2010 Subject Classification}
91G10; 91G20
}
\section{Introduction}\label{sec:Intro}
Fully incomplete markets were introduced in \cite{DolinskyNeufeld.16}.  Roughly speaking, a financial market is fully incomplete if for any volatility process $\alpha$ one can find an equivalent local martingale measure $\Q$ under which $\alpha$ is close to the volatility process $\nu$ of the price process $S$. It turns out that it is a natural appearance for incomplete markets to be fully incomplete. Indeed, it was shown in \cite{DolinskyNeufeld.16} that stochastic volatility models (with unbounded volatility) like the Heston model \cite{Heston.93}, the Hull--White model \cite{HullWhite.87} and the Scott model \cite{Scott.87}, as well as rough volatility models like the one in \cite{GatheralJaissonRosenbaum.14} where the log-volatility is a fractional Ornstein-Uhlenbeck process are all examples of fully incomplete markets. 

The key property of fully incomplete markets, which is the main result obtained in \cite{DolinskyNeufeld.16}, is the following. When a financial agent is allowed to invest in the risky asset $S$ and statically in up to finite  many liquid options, the classical super-replication price (defined with respect to the given initial law $\P$ of the price process) of a (possibly path-dependent) European option $G(S)$ with $G:C[0,T]\to \R_+$ being uniformly continuous and bounded coincides with the robust super-replication price, where the super-hedging property must hold for any path. This follows from the result proven in \cite{DolinskyNeufeld.16} that for fully incomplete markets, the set of all equivalent local martingale measures are weakly dense in the set of all local martingale measures defined on the continuous path space. For more papers related to robust pricing, in particular to duality results, we refer to \cite{AcciaioBeiglbockPenknerSchachermayer.16,BartlKupperNeufeld.17,BartlKupperProemelTangpi.17,BurzoniFrittelliMaggis.15,DolinskySoner.12,DolinskySoner.15a,GuoTanTouzi.15,Hobson.98,HouObloj.15} to name but a few. 

The goal of this (short) paper is to further investigate the super-replication property in fully incomplete markets in the special case where one
 is only allowed to trade in the risky asset (i.e.\ no liquid options available), and the contingent claim is of Markovian type, i.e.\ of the form $g(S_T)$ where $S_T$  denotes the value of the stock at maturity. The main result of this paper states that in that case, even for unbounded and non-continuous payoff functions $g:\R_+ \to \R_+$, the classical super-replication price of $g(S_T)$ coincides with the (even more robust) buy-and-hold super-replication price where one can only buy (or sell) stocks at the initial time and keeps his position till maturity to super-replicate the sold financial 
 claim.  
 To prove our result, we apply techniques which were developed in \cite{DolinskyNeufeld.16}. 
  
 Our paper is motivated by the result in \cite{CvitanicPhamTouzi.99} stating that
 for 
 stochastic volatility models with unbounded volatility, the classical super-replication price of a Markovian claim $g(S_T)$ coincides with the buy-and-hold super-replication price, even for unbounded and non-continuous payoff functions $g$.  A similar result was obtained in \cite{FreySin.99} for
European Call options only,  but in more general stochastic volatility models including e.g.\ the Heston model, which is not covered in \cite{CvitanicPhamTouzi.99} since there strong regularity conditions were imposed on the coefficients of the SDE for the price process which is not fulfilled by square root models.
 
 Summing up, our contribution is twofold. First, our result enlarges the class of price processes from stochastic volatility models to the richer class of fully incomplete markets for which the buy-and-hold property holds when super-replicating $g(S_T)$. Second, for stochastic volatility models, our main result generalizes \cite{CvitanicPhamTouzi.99,FreySin.99} in the sense that we recover the buy-and-hold property when super-replicating $g(S_T)$ for (even more) general stochastic volatility models than in \cite{FreySin.99} for unbounded non-continuous payoff functions $g$ as in \cite{CvitanicPhamTouzi.99}.
 
 The remainder of this paper is organized as follows. In Section~\ref{sec:setup}, we introduce the setup and state the main theorem of this paper. Then, we provide the proof of the main result in Section~\ref{sec:proof}.
\section{Setup and Main Result}\label{sec:setup}
Let $T$ be a finite time horizon and $(\Omega,\mathcal F,\mathbb F, \mathbb P)$ be a filtered probability space, where $\F=\{\mathcal F_t\}_{t=0}^T$ satisfies the usual conditions.  Consider a financial market which consists of one constant bank account $B_t\equiv1$ for all $t\in [0,T]$ and one risky asset with price process 
\begin{equation}\label{eq:def:financial-market}
dS_t= S_t \nu_t\,dW_t, \quad \quad S_0\equiv s_0>0,
\end{equation}
where $\nu=\{\nu_t\}_{t=0}^T$ is an $\mathbb F$-progressively measurable process with given initial point $\nu_0>0$ satisfying $\int_0^T \nu^2_s\,ds<\infty \ \mathbb{P}$-a.s., and  $W=\{W_t\}_{t=0}^T$ denotes a $\P$-$\F$-Brownian motion.

The definition for the financial market \eqref{eq:def:financial-market} to be fully incomplete  was introduced in \cite[Definition~2.1]{DolinskyNeufeld.16}. Let $\mathcal{C}(\nu_0)$ be the set of all
continuous, strictly positive
stochastic processes $\alpha\equiv{\{\alpha_t\}}_{t=0}^T$
which are adapted with respect to the filtration $\mathbb F^W$ generated by $W$ completed by the null sets, and satisfy both
that $\alpha_0=\nu_0$, and  that
$\alpha$ and $\frac{1}{\alpha}$ are uniformly bounded. 
%
%
%
%
\begin{definition}\label{def:fully-incomplete}
The financial market introduced in \eqref{eq:def:financial-market} is
called fully incomplete if
for any $\epsilon>0$
and any process $\alpha\in \mathcal C(\nu_0)$ there exists a
probability measure $\mathbb Q\ll\mathbb P$ such that ${\{W_t\}}_{t=0}^T$ is a $\Q$-$\mathbb F$-Brownian motion and
\begin{equation}\label{eq:def:fully-incomplete}
\mathbb Q(\|\alpha-\nu\|_{\infty}>\epsilon)<\epsilon,
\end{equation}
where $\|u-v\|_{\infty}:=\sup_{0\leq t\leq T}|u_t-v_t|$ denotes the uniform distance between $u$ and $v$.
\end{definition}
Observe that due to the structure of the financial market in $\eqref{eq:def:financial-market}$, by taking convex conbinations of the form $\lambda \P + (1-\lambda) \Q$,  we see that it is equivalent in Definition~\ref{def:fully-incomplete} to require  $\Q \approx \P$ instead of $\Q\ll \P$.  

It turns out that being fully incomplete is a natural phenomenon for incomplete markets.
More precisely, the following sufficient conditions for being a fully incomplete market were proven in \cite{DolinskyNeufeld.16}.
\begin{proposition}\cite[Proposition~2.3.I]{DolinskyNeufeld.16}:\label{prop:fully-incomplete-example1}
Let the volatility process $\nu$ in \eqref{eq:def:financial-market} be of the form
\begin{equation} \label{eq:ex:stoch-vol}
d\nu_t=a(t,\nu_t)\,dt+b(t,\nu_t) \,d{\widehat W}_t+c(t,\nu_t) \,dW_t, \quad \quad \nu_0>0.
\end{equation}
If the SDE in \eqref{eq:ex:stoch-vol} has a unique strong solution and the solution is strictly positive, and if
the functions $a,b,c: [0,T]\times (0,\infty)\rightarrow\mathbb R$
are continuous and for any $t\in [0,T]$, $x>0$ we have $b(t,x)>0$, 
then 
the corresponding financial market defined in \eqref{eq:def:financial-market} is fully incomplete. 
\end{proposition}
The second sufficient condition proven in \cite{DolinskyNeufeld.16} for being a fully incomplete market is strongly related to the so-called Conditional Full Support (CFS)  property introduced in \cite{guasoni2008consistent}.  A stochastic process $\Sigma=\{\Sigma_t\}_{t=0}^T$ has the Conditional Full Support (CFS)  property if
for all $t\in (0,T]$
\begin{equation}
\label{eq:CFS}
\mbox{supp} \ \mathbb P(\Sigma_{|[t,T]}|\Sigma_{|[0,t]})=C_{\Sigma_t}[t,T] \ \ \mbox{a.s.,}
\end{equation}
where
$C_y[t,T]$ is the set of all continuous functions $f:[t,T]\rightarrow\mathbb{R}_{+}$
with $f(t)=y$. 
It means that
from any given time on, $\Sigma$ can continue arbitrarily close
to any given path with positive conditional probability.
\begin{proposition}\cite[Proposition~2.3.II]{DolinskyNeufeld.16}:\label{prop:fully-incomplete-example2}
Let the filtration $\mathbb F$ 
 here
be the usual augmentation of the filtration generated by $W$ and $\nu$. If 
$\nu_t=\nu^{(1)}_t\nu^{(2)}_t$, where
$\nu^{(1)}$ is adapted to the filtration generated by $W$,
$\nu^{(2)}$ is independent of $W$, both
$\nu^{(1)}$, $\nu^{(2)}$ are strictly positive and continuous processes, 
and
 $\ln\nu^{(2)}$ has the (CFS) property,
then
the financial market given by \eqref{eq:def:financial-market} is fully incomplete.
\end{proposition}
As a consequence of Proposition~\ref{prop:fully-incomplete-example1} \& \ref{prop:fully-incomplete-example2}, popular stochastic volatility models like the Heston model, Hull--White model and the Scott model, as well as rough volatility models like the one in \cite{GatheralJaissonRosenbaum.14} are examples of fully incomplete markets; we refer to \cite{DolinskyNeufeld.16} for more details.


To recall the main property of fully incomplete markets, let $\mathfrak A^\P$ denote the set of all $\mathbb F$-progressively measurable processes $\{\gamma_t\}_{t=0}^T$ with $\int_0^T \gamma^2_t\nu^2_t S^2_t\,dt<\infty \ \P$-a.s. such that the stochastic integral $\int \gamma\,dS$ is uniformly bounded from below. The robust setup is defined as follows. 
Let $\{\mathbb S_t\}_{t=0}^T$ be the canonical process on the space $C[0,T]$, i.e.\
$\mathbb S_t(\omega)=\omega(t)$, $\omega\in C[0,T]$, and
$\mathbb F^\mathbb S_t=\sigma\{\mathbb S_u: u\leq t\}, \ t\in [0,T]$ denotes the canonical filtration. The set $\mathfrak A$ consists of processes ${\{\gamma_t\}}_{t=0}^T$ which are $\mathbb F^\mathbb S$-adapted and
of bounded variation with left-continuous paths such that the process
$\int \gamma d\mathbb S$ is uniformly bounded from below, where here $\int \gamma d\mathbb S$ is defined by
\begin{equation}\label{eq:stoch-integral}
\int_{0}^t \gamma_u d\mathbb S_u:=\gamma_t\mathbb S_t-\gamma_0\mathbb S_0-\int_{0}^t \mathbb S_s d\gamma_s,  \quad t\in [0,T],
\end{equation}
using the standard Lebesgue-Stieltjes integral for the last integral. Finally, define $\mathfrak S$ to be  the set of all paths in $C[0,T]$ which are strictly positive and starts in $S_0$. Then the main theorem in \cite{DolinskyNeufeld.16} states the following.
\begin{theorem}\cite[Theorem~3.1]{DolinskyNeufeld.16}:\label{thm:Fully-Incomplete-Dolinsky-Neufeld}
Let $G:C[0,T]\rightarrow \mathbb{R}$ be a bounded and uniformly continuous function and consider the (path-dependent) European option $G(S)$. Moreover, for every $i:=1,\dots, N$ let  $h_i:C[0,T]\to \R$ be a bounded and uniformly continuous function and $h_i(S)$ be a static position with price $\mathcal P_i$. If the financial market defined in \eqref{eq:def:financial-market} is fully incomplete,
%
then the classical super-replication price
\begin{equation}
\begin{split}
&V^\P_{h_i,\dots,h_N}(G)\\
&:=\inf\bigg\{ x+ \sum_{i=1}^N c_i \mathcal P_i \,\bigg|\, \exists \ (\gamma_t) \in \mathfrak A^\P \mbox{ s.t.\ \ } x + \sum_{i=1}^N c_i h_i(S) + \int_0^T \gamma_t \,dS_t \geq G(S)\  \P\mbox{-a.s.}\bigg\}
\end{split}
\label{eq:P-price}
\end{equation}
coincides with the robust super-replication price
\begin{equation}
\begin{split}
&V_{h_i,\dots,h_N}(G)\\
&:=\inf\bigg\{ x+ \sum_{i=1}^N c_i \mathcal P_i \,\bigg|\, \exists \ (\gamma_t) \in \mathfrak A \mbox{ s.t.\ \ } x + \sum_{i=1}^N c_i h_i(\mathbb S) + \int_0^T \gamma_t \,d\mathbb S_t \geq G(\mathbb S)\ \forall \mathbb S \in \mathfrak S\bigg\},
\end{split}
\label{eq:robust-price}
\end{equation}
i.e.\ we have
 $V^\P_{h_i,\dots,h_N}(G)= V_{h_i,\dots,h_N}(G).$
\end{theorem}
In this paper, we analyze the super-replication property when the financial market defined in \eqref{eq:def:financial-market} is fully incomplete in the special case where  the option is of Markovian type, i.e. the option is of the form $g(S_T)$ for some function $g:\R_+\to \R_+$, and  there are no liquid options to trade with. In this case, the classical super-replication price of $g(S_T)$ is given by
\begin{equation}\label{eq:P-price-0}
V^\P_0(g):=\inf\bigg\{ x \in \R\,\bigg|\, \exists \ (\gamma_t) \in \mathfrak A^\P \mbox{ s.t. \ } x +  \int_0^T \gamma_t \,dS_t \geq g(S_T)\  \P\mbox{-a.s.}\bigg\}.
\end{equation}
The robust price $V_0(g)$ is defined analogously. Another (even more robust) super-replication price 
is the one where only buy-and-hold strategies are allowed.
Its formal definition is
\begin{equation}\label{eq:P-price-0-B-H}
V^{B\&H}_0(g):=\inf\bigg\{ x \in \R \,\bigg|\, \exists \ \Delta \in \R  \mbox{ s.t. \ } x +  \Delta (S_T(\omega)-S_0) \geq g(S_T(\omega))\ \forall \omega \in \Omega\bigg\}.
\end{equation}
Clearly, for any option $G$ in any financial market 
\begin{equation}\label{eq:price-relation-easy}
V^\P_0(G) \leq V_0(G) \leq V^{B\&H}_0(G),
\end{equation}
and by \cite{DolinskyNeufeld.16} we know that in fully incomplete markets, a priori, $V^\P_0(G) = V_0(G)\leq V^{B\&H}_0(G)$ holds true for  options $G$ which are bounded and uniformly continuous. 

The goal of this paper is to show that for Markovian claims, the above inequalities are in fact true equalities, even for unbounded and non-continuous payoff functions. Moreover, the price and the optimal (buy-and-hold) strategy can be calculated explicitly.
\begin{theorem}\label{thm:Buy-and-Hold}
Let $g:\R_{+}\to \R_+$ be a nonnegative and lower semicontinuous function. If the financial market defined in \eqref{eq:def:financial-market} is fully incomplete, then	
%
the super-replication price of the  Markovian claim  $g(S_T)$ satisfies 
\begin{equation}\label{eq:thm:Buy-and-Hold-price}
V^\P_0(g)=V^{B\&H}_0(g)= \widehat g(S_0), 
\end{equation}
where $\widehat g$ denotes the concave envelope of $g$.
Moreover, an optimal (buy-and-hold) strategy exists and is explicitly defined by
\begin{equation}\label{eq:thm:Buy-and-Hold-strategy}
\gamma\equiv \partial_+ \widehat g(S_0).
\end{equation}
\end{theorem}
\begin{remark}
By the cash-invariance property of both $V^\P_0$ and $V^{B\&H}_0$,  the condition that $g:\R_{+}\to \R_+$ is nonnegative could be relaxed by the requirement to be bounded from below.
\end{remark}
Next we provide an example to show that Theorem~\ref{thm:Buy-and-Hold} 
may already fail when there is one liquid option to statically trade with.
%
\begin{example}\label{ex:Counter-ex}
Let the financial market defined in \eqref{eq:def:financial-market} be fully incomplete, let $K>0$, and consider the function $g:\R_+ \to \R_+$ defined by $g(x)=(x-K)^+$, $x \in \R_+$. Notice that the function $g$ satisfies the conditions imposed in Theorem~\ref{thm:Buy-and-Hold} and the claim $g(S_T)=(S_T-K)^+$ corresponds to the European Call option with strike $K$ and maturity $T$. Assume that one can statically trade the European Put option $h(S_T)=(K-S_T)^+$ (with the same strike $K$ and maturity $T$), where the given price $\mathcal{P}$ of $h(S_T)$ satisfies $\mathcal{P}=\E^{\Q}[h(S_T)]$ for some equivalent local martingale measure $\Q$. Then, we have that
\begin{align}
&V^{B\&H}_h(g)\nonumber\\
&:=\inf\bigg\{ x + c\mathcal{P}  \,\bigg|\, \exists \ \Delta \in \R  \mbox{ s.t. \ } x + c h(S_T(\omega))+  \Delta (S_T(\omega)-S_0) \geq g(S_T(\omega))\ \forall \omega \in \Omega\bigg\} \nonumber\\
&<V^{B\&H}_0(g). \label{eq:example-B-H}
\end{align} 
We provide its proof at the end of Section~\ref{sec:proof}.
\end{example}
%
%
\section{Proof of Theorem~\ref{thm:Buy-and-Hold} and Example~\ref{ex:Counter-ex}}\label{sec:proof}
We start the proof by first recalling the well-known (trivial) inequalities in \eqref{eq:thm:Buy-and-Hold-price}, namely:
\begin{lemma}\label{le:easy-ineq}
Let $g\colon \R_+ \to \R$ be a function and $g(S_T)$ be the corresponding Markovian claim. Then
we have that
\begin{equation}\label{le:eq:easy-ineq}
V^\P_0(g)  \leq V^{B\&H}_0(g)\leq \widehat g(S_0),
\end{equation}
and for $\gamma \equiv\partial_+ \widehat g(S_0)$, the pair $(\widehat g(S_0),\gamma)$ is a buy-and-hold super-replicating strategy, i.e.
\begin{equation}\label{le:eq:easy-superhedge}
\widehat g(S_0) + \partial_+ \widehat g(S_0)\, (S_T(\omega)-S_0) \geq g(S_T(\omega)) \quad \forall \omega \in \Omega.
\end{equation}
\end{lemma}
\begin{proof}
Clearly, $V^\P_0(g)  \leq V^{B\&H}_0(g)$. Moreover, by the definition of $\widehat g$ being the smallest concave function bigger than $g$, we obtain for any $\omega \in \Omega$ that
\begin{equation}\label{le:eq:easy-proof}
\widehat g(S_0) + \partial_+ \widehat g(S_0)\, (S_T(\omega)-S_0) \geq \widehat g(S_T(\omega)) \geq g(S_T(\omega)).
\end{equation}
\end{proof}
%
%
%
%
Next, we show that it is sufficient to prove the results in Theorem~\ref{thm:Buy-and-Hold} for any bounded, nonnegative payoff function $g:\R_{+}\to \R_+$ which is Lipschitz continuous.
Beforehand, let us quickly introduce the following notion, which we will use frequently in the rest of this paper. For any $x>0$ and any (sufficiently integrable) progressively
measurable process  $\alpha=\{\alpha_t\}_{t=0}^T$
we denote the corresponding stochastic exponential with respect to $W$ by
\begin{equation}\label{eq:def-S-alpha}
\mathbf S^{\alpha,x}_t:=x\,e^{\int_{0}^t \alpha_v\, dW_v-\frac{1}{2}\int_{0}^t \alpha^2_v\, dv}, \ \ t\in [0,T].
\end{equation}

\begin{lemma}\label{le:reduction-Lipschitz}
If  \eqref{eq:thm:Buy-and-Hold-price} and \eqref {eq:thm:Buy-and-Hold-strategy} in Theorem~\ref{thm:Buy-and-Hold} hold true for any bounded, nonnegative 
function $g:\R_{+}\to \R_+$ which is Lipschitz continuous, then \eqref{eq:thm:Buy-and-Hold-price} and \eqref {eq:thm:Buy-and-Hold-strategy} also hold true for any nonnegative 
function $g:\R_{+}\to \R_+$ which is  lower semicontinuous.
\end{lemma}
\begin{proof}
Let $g:\R_{+}\to \R_+$ be any nonnegative, lower semicontinuous function. Due to Lemma~\ref{le:easy-ineq}, it remains to show that $V^\P_0(g) \geq \widehat g(S_0)$.
  Define the sequence of functions
\begin{equation}
\begin{split}
\widetilde g_n(x) &:= \inf_{y\geq 0} \big\{g(y)+ n|x-y|\big\}, \quad x\geq 0;\\
 g_n(x) &:= \min\big\{\widetilde g_n(x),n \big\}, \quad x\geq 0.
 \end{split}
 \label{eq:Lipschitz-approx}
\end{equation}
We see that for each $n$, the function $g_n$ is bounded, nonnegative and Lipschitz continuous (with Lipschitz constant $n$). Moreover, the sequence $(g_n)$ converges non-decreasingly to $g$. Let $\mathcal T$ denote all $\F$-stopping times. 
Using \cite[Lemma~5.4]{CvitanicPhamTouzi.99} and the monotone convergence theorem, we see that
\begin{equation}\label{eq:le:reduction-Lipschitz2}
\widehat g(S_0)
= \sup_{\tau \in \mathcal T} \E^\P[g(\mathbf S^{1,S_0}_\tau)]
=\sup_{n \in \N} \sup_{\tau \in \mathcal T}  \E^\P[g_n(\mathbf S^{1,S_0}_\tau)]
= \sup_{n \in \N} \widehat g_n(S_0)
\end{equation}
By the assumption that \eqref{eq:thm:Buy-and-Hold-price} and \eqref {eq:thm:Buy-and-Hold-strategy} holds true for bounded, nonnegative payoff functions which are Lipschitz continuous and since the super-replication price is monotone in the claim 
\begin{equation} \label{eq:le:reduction-Lipschitz3}
\sup_{n \in \N} \widehat g_n(S_0)=\sup_{n \in \N}  V^\P_0(g_n) \leq V^\P_0(g) 
\end{equation}
Thus, by \eqref{eq:le:reduction-Lipschitz2}--\eqref{eq:le:reduction-Lipschitz3}  we obtained $\widehat g(S_0)\leq V^\P_0(g)$ as desired.
\end{proof}
%
%
Due to Lemma \ref{le:reduction-Lipschitz}, it suffices to prove Theorem~\ref{thm:Buy-and-Hold} for bounded, Lipschitz continuous payoff functions $g: \R_+ \to \R_+$. To do so, we first start with a Lemma regarding an upper bound for the concave envelope.
%
%
\begin{lemma}\label{le:upper-bound-envelope}
Let $g:\R_+ \to \R_+$ be bounded, nonnegative and Lipschitz continuous. Then
\begin{equation}\label{le:eq:upper-bound-envelope}
\widehat{g}(s) \leq \sup_{\alpha \in \mathcal C(\nu_0)} \E^\P[g(\mathbf S^{\alpha,s}_T)], \quad \quad \forall s> 0.
\end{equation}
\end{lemma}
\begin{proof}
Introduce the set $\mathcal{A}$ of all nonnegative progressively measurable processes with respect to the  filtration $\mathbb F^W$ generated by the  Brownian motion $W$ satisfying $\int_0^T \alpha^2_t\,dt <\infty $ $\P$-a.s. and for which there exists a constant $C>0$ (which may depend on $\alpha$), such that $\frac{1}{C}\leq \mathbf S^{\alpha,1}\leq C$. By the same argument as in \cite[Lemma~7.1]{DolinskyNeufeld.16}, the function $\mathfrak G:(0,\infty)\to \R$
defined by
\begin{equation}\label{le:upper-bound-envelope-pf1}
\mathfrak G(s):= \sup_{\alpha \in \mathcal A} \E^\P\big[ g(\mathbf S^{\alpha,s}_T)\big]
\end{equation}
is concave and satisfies $g\leq \mathfrak G$. By the minimality property of the concave envelope of $g$, this means that also 
\begin{equation}\label{le:upper-bound-envelope-pf2}
\widehat g(s) \leq \mathfrak G(s), \quad \forall s>0.
\end{equation} 
Then, by the same approximation argument as in \cite[Lemma~7.2]{DolinskyNeufeld.16}, where in this step we use the Lipschitz property of $g$, we obtain that
\begin{equation}\label{le:upper-bound-envelope-pf3}
\mathfrak G(s) \leq \sup_{\alpha \in \mathcal C(\nu_0)} \E^\P[g(\mathbf S^{\alpha,s}_T)],   \quad  \forall s>0,
\end{equation}
which implies the result.
\end{proof}

%
%
Now we are able to finish the proof of Theorem~\ref{thm:Buy-and-Hold}.
\begin{proof}[Proof of Theorem~\ref{thm:Buy-and-Hold}]
We follow a similar argument as the one in \cite[Theorem~4.2]{DolinskyNeufeld.16}.
Let $g: \R_+ \to \R_+$ be bounded, nonnegative and Lipschitz continuous. By Lemma~\ref{le:easy-ineq} and Lemma~\ref{le:reduction-Lipschitz}, it remains to show that $V_0^\P(g)\geq \widehat g(S_0)$. Thus, we fix any $x>V_0^\P(g)$ and need to show that $x \geq \widehat g(S_0)$. 
By definition of $x$, there exists an $\F$-predictable, $S$-integrable process $\{\gamma_t\}_{t=0}^T$ such that
\begin{equation}\label{eq:pf-supehed}
x + \int_0^T \gamma_t \, d S_t \geq g(S_T) \quad \P\mbox{-a.s.}
\end{equation}
Fix any $\varepsilon>0$.
By Lemma~\ref{le:upper-bound-envelope}, there exists $\alpha \in \mathcal C (\nu_0)$ such that
\begin{equation}\label{eq:proof-thm-buy-and-hold-1}
\widehat g(S_0)-\varepsilon\leq \E^\P\big[g(\mathbf S^{\alpha,S_0}_T)\big].
\end{equation}
Next, choose any $\delta>0$ small enough (i.e. $\delta \ll \varepsilon$). Since the financial market defined in \eqref{eq:def:financial-market} is fully incomplete, there exists $\Q\ll\P$ such that $W$ is a $\Q$-$\F$-Brownian motion and
\begin{equation}\label{eq:proof-thm-buy-and-hold-FULLY}
\Q\big( \Vert \alpha-\nu\Vert_\infty\geq\delta\big)<\delta.
\end{equation}
Since $(x,\gamma)$ is a super-replicating strategy, the supermartingale property of the gain process yields  
\begin{align}
\E^\Q\big[g(S_T)\big]
\leq \E^\Q\Big[x + \int_0^T \gamma_t \,dS_t\Big] \leq x.
 \label{eq:proof-thm-buy-and-hold-2}
\end{align}
Now, define the stopping time
\begin{equation}\label{eq:pf-def-stopping-time}
\tau:= \inf \big\{ t \geq 0 \, \big|\, |\alpha_t-\nu_t|\geq\delta \big\} \wedge T.
\end{equation}
Then, by definition of $\tau$
\begin{equation}
\begin{split}
\frac{1}{2} \int_0^\tau |\alpha^2_t-\nu^2_t| \, dt 
&=  \frac{1}{2} \int_0^\tau |\alpha^2_t-(\alpha_t-(\alpha_t-\nu_t))^2| \, dt\\
&= \frac{1}{2} \int_0^\tau |\alpha^2_t-(\alpha_t^2-2 \alpha_t (\alpha_t-\nu_t) +(\alpha_t-\nu_t)^2)| \, dt\\
&= \frac{1}{2} \int_0^\tau |2 \alpha_t (\alpha_t-\nu_t) -(\alpha_t-\nu_t)^2)| \, dt \\
&\leq \frac{1}{2} T \delta \,(2 \Vert \alpha \Vert_\infty + \delta),
\end{split}
\label{eq:pf-stopping-ineq1}
\end{equation}
as well as due to the It\^o isometry
\begin{equation}\label{eq:pf-ineq-Ito-isometry}
\E^\Q\Big[\Big(\int_0^\tau(\alpha_t-\nu_t) \, dW_t\Big)^2\Big]\leq \delta^2 T.
\end{equation}
Thus, Chebyshev's inequality implies that
\begin{equation}
\begin{split}
&\Q\Big( \frac{1}{2} \int_0^\tau | \alpha_t^2- \nu_t^2|\,dt + \Big|\int_0^\tau (\alpha_t-\nu_t)\, dW_t\Big| \geq 2 \sqrt{\delta}\Big) \\
&\leq 
\frac{\frac{1}{2} \E^\Q\big[\int_0^\tau | \alpha_t^2- \nu_t^2|\,dt\big]}{\sqrt{\delta}} 
+ \frac{\E^\Q\big[(\int_0^\tau(\alpha_t-\nu_t) \, dW_t)^2\big]}{\delta}\\
&\leq 
\frac{\frac{1}{2}  T \delta \,(2 \Vert \alpha \Vert_\infty + \delta)}{\sqrt{\delta}} + \frac{\delta^2 T}{\delta}\\
&\leq
c(\sqrt{\delta}+ \delta). 
\end{split}
\label{eq:proof-thm-buy-and-hold-3}
\end{equation}
for some constant $c$ which may depend on $\varepsilon$ (but not on $\delta$). Now, since by definition, the price process $S$ is defined as $S\equiv\mathbf S^{\nu,S_0}$, we deduce from \eqref{eq:proof-thm-buy-and-hold-3} that for sufficiently small $\delta$
\begin{equation}
\begin{split}
\Q\Big(|\ln S_{\tau}- \ln \mathbf S^{\alpha,S_0}_\tau|> \varepsilon\Big)
&\leq 
\Q\Big(\frac{1}{2} \int_0^\tau | \alpha_t^2- \nu_t^2|\,dt + \Big|\int_0^\tau (\alpha_t-\nu_t)\, dW_t\Big| > \varepsilon\Big)\\
&\leq c(\sqrt{\delta}+ \delta).
\end{split}
\label{eq:proof-thm-buy-and-hold-4}
\end{equation}
Next, define the event
\begin{equation}\label{eq:pf-event}
U_\varepsilon:= \{\tau<T\} \cup \{|\ln S_{\tau}- \ln \mathbf S^{\alpha,S_0}_\tau|> \varepsilon\}
\end{equation}
It is elementary to check that since $g\colon\R_+ \to \R_+$ is Lipschitz continuous with respect to some Lipschitz-constant $K>0$, we have for any $x>0, y>0$ with $|\ln(x)-\ln(y)|\leq \epsilon$, that $g(x)\geq  g(y)-Ky\,(e^{\epsilon}-1)$. Therefore, using \eqref{eq:proof-thm-buy-and-hold-2} and the definition of (the complement of) $U_\varepsilon$ yields
\begin{equation}
\begin{split}
x\geq  \E^\Q\big[g(S_T)\big]
&\geq \E^\Q\big[\mathbf{1}_{U^c_\epsilon} \,g(S_T)\big] \\
 &\geq 
 \E^\Q\big[\,g(\mathbf S^{\alpha,S_0}_T)\big] - K(e^{\epsilon}-1) \,\E^\Q\big[\mathbf S^{\alpha,S_0}_T\big]
- \E^\Q\big[\mathbf 1_{U_\epsilon} \, g(\mathbf S^{\alpha,S_0}_T)\big]  \\
&\geq 
 \E^\Q\big[\,g(\mathbf S^{\alpha,S_0}_T)\big] - K(e^{\epsilon}-1) S_0
- \E^\Q\big[\mathbf 1_{U_\epsilon} \, g(\mathbf S^{\alpha,S_0}_T)\big].
\end{split}
 \label{eq:proof-thm-buy-and-hold-5}
 \end{equation}
Now, before letting $\varepsilon$ go to zero, we first need to analyze $\E^\Q[\mathbf 1_{U_\epsilon} \, g(\mathbf S^{\alpha,S_0}_T)]$. To that end, see that by the Lipschitz property of $g\geq0$ 
\begin{equation}
\begin{split}
\E^\Q\big[g(\mathbf S^{\alpha,S_0}_T)^2\big] \leq \E^\Q\big[(g(0) + K \mathbf S^{\alpha,S_0}_T)^2\big]
&\leq 2 g(0)^2 + 2K^2\, 
\E^\Q\big[(\mathbf S^{\alpha,S_0}_T)^2\big].
 \end{split}
 \label{eq:pf-g-square1}
\end{equation}
Moreover, since $\alpha \in \mathcal C(\nu_0)$ is uniformly bounded, we see that for some constant $C>0$, 
\begin{equation} \label{eq:pf-g-square2}
\E^\Q\big[(\mathbf S^{\alpha,S_0}_T)^2\big]\leq C \E^\Q\big[\mathbf S^{2\alpha,S_0}_T\big]<\infty.
\end{equation}
Hence
we conclude that $\E^\Q[g(\mathbf S^{\alpha,S_0}_T)^2]<\infty$.
Furthermore, observe that 
\begin{equation} \label{eq:pf-event-ineq}
\{\tau<T\}\subseteq \{\Vert \alpha-\nu\Vert_\infty\geq\delta\}.
\end{equation} By the Cauchy-Schwarz inequality, \eqref{eq:proof-thm-buy-and-hold-FULLY} and \eqref{eq:proof-thm-buy-and-hold-4}, this implies for sufficiently small $\delta$ that
\begin{equation}\label{eq:pf-event-cond-ineq}
\E^\Q\big[\mathbf 1_{U_\varepsilon} \, g(\mathbf S^{\alpha,S_0}_T)\big]\leq \E^\Q\big[g(\mathbf S^{\alpha,S_0}_T)^2\big]^{\frac{1}{2}} \, \Q\big[U_\epsilon\big]^{\frac{1}{2}} \leq C (\delta + \sqrt{\delta})^{\frac{1}{2}}<\varepsilon.
\end{equation}
Therefore, we deduce from \eqref{eq:proof-thm-buy-and-hold-5} that
\begin{equation}\label{eq:pf-x-lower-bound1}
x\geq \E^\Q\big[\,g(\mathbf S^{\alpha,S_0}_T)\big] - K(e^{\epsilon}-1) S_0
- \varepsilon.
\end{equation}
Recall that under both measures $\P$ and $\Q$, the process $W$ is an $\F$-Brownian motion, and that the process $\alpha$ is  progressively measurable with respect to the smaller filtration $\mathbb F^W$ generated by $W$. In particular, $\alpha= \varphi(W)$ for some progressively measurable map $C[0,T]\to C[0,T]$. Therefore, the law of $\mathbf S^{\alpha,S_0}_T$ under $\P$ and $\Q$ are the same. This implies together with \eqref{eq:proof-thm-buy-and-hold-1} that
\begin{equation}
\begin{split}
x&\geq \E^\Q\big[\,g(\mathbf S^{\alpha,S_0}_T)\big] - K(e^{\epsilon}-1)S_0 - \varepsilon\\
&= \E^\P\big[\,g(\mathbf S^{\alpha,S_0}_T)\big] - K(e^{\epsilon}-1)S_0 -  \varepsilon \\
&\geq \widehat g(S_0)- K(e^{\epsilon}-1)S_0 -  2\varepsilon.
\end{split}
\label{eq:pf-x-lower-bound2}
\end{equation}
Since $\varepsilon>0$ was arbitrarily chosen, we can now let $\epsilon$ go to zero to obtain desired inequality 
\begin{equation}\label{eq:pf-x-lower-bound-desired}
x\geq \widehat g(S_0).
\end{equation}
\end{proof}
We finish this section with the proof of Example~\ref{ex:Counter-ex}.
\begin{proof}[Proof of Example~\ref{ex:Counter-ex}]
First, notice that since by assumption the financial market defined in \eqref{eq:def:financial-market} is fully incomplete and the function $g(x):=(x-K)^+$, $x \in \R_+$, is nonnegative and continuous, we deduce from Theorem~\ref{thm:Buy-and-Hold} that
\begin{equation}\label{ex:V0}
V^{B\&H}_0(g)= \widehat g(S_0)=S_0.
\end{equation}
Therefore, it remains to show that $V^{B\&H}_h(g)<S_0$ for 
$h(S_T):=(K-S_T)^+$. 
To that end, 
notice that for any equivalent local martingale measure $\Q$ we have that $\E^{\Q}[(K-S_T)^+]<K$. Indeed, since $S$ is strictly positive, we have  $(K-S_T)^+< K$ and so also $\E^{\Q}[(K-S_T)^+]< K$. 

Now, 
recall that by assumption, the price of the static option $h(S_T)=(K-S_T)^+$ is equal to $\cP=\E^\Q[(K-S_T)^+]$ for some equivalent local martingale measure $\Q$, hence from the above argument we know that $\cP<K$. Next observe that starting with initial capital $S_0-K$, buying one stock and holding this position as well as buying one European Put option (super-)replicates pointwise the European Call option.
Indeed, for any $\omega \in \Omega$, we have that
\begin{equation}
\label{ex:superrep1}
(S_0-K) + (K-S_T(\omega))^+ + (S_T(\omega)-S_0)
=(S_T(\omega)-K)^+.
\end{equation}
Therefore, using $x:=S_0-K$, $c:=1$,  $\Delta:=1$ we see that 
\begin{equation}
\label{ex:superrep2}
x+ ch(S_T(\omega)) + \Delta(S_T(\omega)-S_0)\geq g(S_T(\omega)) \quad \forall \omega \in \Omega,
\end{equation}
which by definition of $V^{B\&H}_h(g)$ and the above proven fact that $\cP<K$ ensure that
\begin{equation}\label{ex:superrep-conclusion}
V^{B\&H}_h(g)\leq x+ c\mathcal{P}=(S_0-K) + \cP < S_0. 
\end{equation}
 	\end{proof}
\bibliographystyle{plainnat}  
\newcommand{\dummy}[1]{}


\end{document}